\documentclass[copyright,creativecommons,noderivs,10pt]{eptcs}
\providecommand{\event}{QPL 2013}
\usepackage{breakurl}

\usepackage{color}
\usepackage{units,amsmath,amsfonts,amssymb,xspace,stmaryrd}
\usepackage{mathtools}
\usepackage{diagrams}
\usepackage{psfrag}

\usepackage{amsmath}
\usepackage[mathscr]{eucal}

\newcommand{\ket}[1]{\left| #1 \right\rangle}

\newcommand{\CZ}{\ensuremath{\wedge Z}\xspace}

\newcommand{\fdhilb}{\bf FdHilb}
\newcommand{\wfdhilb}{\fdhilb_{\text{wp}}}
\newcommand{\catD}{\ensuremath{{\cal D}}\xspace}

\newcommand{\denote}[1]{\llbracket #1 \rrbracket} 
\newcommand{\ldenote}[1]{\left\llbracket #1 \right\rrbracket}

\newcommand{\zxcalculus}{\textsc{zx}-calculus\xspace}

\usepackage{amsthm}
\newtheorem{theorem}{Theorem}[section]
\newtheorem{proposition}[theorem]{Proposition}
\newtheorem{lemma}[theorem]{Lemma}

\theoremstyle{definition}\newtheorem{example}[theorem]{Example}
\theoremstyle{definition}
\theoremstyle{definition}\newtheorem{definition}[theorem]{Definition}
\theoremstyle{definition}

\newtheorem*{remark}{Remark}

{\bfseries}{\itshape}

\usepackage{tikz} 

\pgfdeclarelayer{edgelayer}
\pgfdeclarelayer{nodelayer}
\pgfsetlayers{edgelayer,nodelayer,main}

\tikzstyle{greendot}=[circle,fill=green,draw=black,inner sep=2pt,minimum size=0.20cm]
\tikzstyle{reddot}=[circle,fill=red,draw=black,inner sep=2pt,minimum size=0.20cm]
\tikzstyle{greenrot}=[circle,fill=green,draw=black,inner
sep=2pt,minimum size=0.20cm]
\tikzstyle{redrot}=[circle,fill=red,draw=black,inner sep=2pt,minimum size=0.20cm]
\tikzstyle{blackrot}=[circle,fill=black,draw=black,inner sep=1.5pt,minimum size=0.10cm]

\tikzstyle{box}=[rectangle,fill=white,draw=black,text centered,inner sep=0pt,minimum size=0.15cm]
\tikzstyle{empty}=[circle,fill=white,draw=white,inner sep=0pt,minimum size=0cm]
\tikzstyle{none}=[]

\newcommand{\smdiag}[1]{\vcenter{\hbox{#1}}}
    
\newcommand{\greenspider}[1]{\smdiag{%
  \begin{tikzpicture}
    \begin{pgfonlayer}{nodelayer}
      \node [style=greenrot] (0) at (0, 0) {$#1$};
      \node [style=empty] (1) at (-0.5, 0.5) {};
      \node [style=empty] (2) at (0.5, 0.5) {};
      \node [style=empty] (3) at (-0.5, -0.5) {};
      \node [style=empty] (4) at (0.5, -0.5) {};
      \node [style=empty] (5) at (0, 0.5) {$\cdots$};
      \node [style=empty] (6) at (0, -0.5) {$\cdots$};
    \end{pgfonlayer}
    \begin{pgfonlayer}{edgelayer}
      \draw (0) to (1);
      \draw (2) to (0);
      \draw [] (0) to (3);
      \draw [] (0) to (4);
    \end{pgfonlayer}
  \end{tikzpicture}
}}

\newcommand{\redspider}[1]{\smdiag{%
  \begin{tikzpicture}
    \begin{pgfonlayer}{nodelayer}
      \node [style=redrot] (0) at (0, 0) {$#1$};
      \node [style=empty] (1) at (-0.5, 0.5) {};
      \node [style=empty] (2) at (0.5, 0.5) {};
      \node [style=empty] (3) at (-0.5, -0.5) {};
      \node [style=empty] (4) at (0.5, -0.5) {};
      \node [style=empty] (5) at (0, 0.5) {$\cdots$};
      \node [style=empty] (6) at (0, -0.5) {$\cdots$};
    \end{pgfonlayer}
    \begin{pgfonlayer}{edgelayer}
      \draw (0) to (1);
      \draw (2) to (0);
      \draw [] (0) to (3);
      \draw [] (0) to (4);
    \end{pgfonlayer}
  \end{tikzpicture}
}}

\newcommand{\RULEgreenspiderlhs}{\smdiag{%
\begin{tikzpicture}
  \begin{pgfonlayer}{nodelayer}
    \node [style=greenrot] (0) at (0, 0.5) {$\alpha$};
    \node [style=greenrot] (1) at (0, -0.5) {$\beta$};
    \node [style=none] (2) at (-0.5, 1.25) {};
    \node [style=none] (3) at (0, 1.25) {$\cdots$};
    \node [style=none] (4) at (0.5, 1.25) {};
    \node [style=none] (5) at (-0.5, -1.25) {};
    \node [style=none] (6) at (0, -1.25) {$\cdots$};
    \node [style=none] (7) at (0.5, -1.25) {};
  \end{pgfonlayer}
  \begin{pgfonlayer}{edgelayer}
    \draw (2.center) to (0);
    \draw (4.center) to (0);
    \draw (0) to (1);
    \draw (1) to (5.center);
    \draw (1) to (7.center);
  \end{pgfonlayer}
\end{tikzpicture}
}}

\newcommand{\RULEgreenspiderrhs}{\smdiag{%
\begin{tikzpicture}
  \begin{pgfonlayer}{nodelayer}
    \node [style=greenrot] (0) at (0, 0) {$\alpha\!\!+\!\!\beta$};
    \node [style=none] (1) at (-0.5, 1.25) {};
    \node [style=none] (2) at (0, 1.25) {$\cdots$};
    \node [style=none] (3) at (0.5, 1.25) {};
    \node [style=none] (4) at (-0.5, -1.25) {};
    \node [style=none] (5) at (0, -1.25) {$\cdots$};
    \node [style=none] (6) at (0.5, -1.25) {};
  \end{pgfonlayer}
  \begin{pgfonlayer}{edgelayer}
    \draw (1.center) to (0);
    \draw (3.center) to (0);
    \draw (4.center) to (0);
    \draw (0) to (6.center);
  \end{pgfonlayer}
\end{tikzpicture}
}}

\newcommand{\RULEantilooplhs}{\smdiag{%
\begin{tikzpicture}
  \begin{pgfonlayer}{nodelayer}
    \node [style=greenrot] (0) at (0, 0) {$\alpha$};
    \node [style=none] (1) at (-1, 0.75) {};
    \node [style=none] (2) at (-0.5, 0.75) {$\cdots$};
    \node [style=none] (3) at (0, 0.75) {};
    \node [style=none] (4) at (-1, -0.75) {};
    \node [style=none] (5) at (-0.5, -0.75) {$\cdots$};
    \node [style=none] (6) at (0, -0.75) {};
  \end{pgfonlayer}
  \begin{pgfonlayer}{edgelayer}
    \draw (1.center) to (0);
    \draw (3.center) to (0);
    \draw (4.center) to (0);
    \draw (0) to (6.center);
    \draw [in=45, out=315, loop] (0) to ();
  \end{pgfonlayer}
\end{tikzpicture}
}}

\newcommand{\RULEantilooprhs}{\smdiag{%
\begin{tikzpicture}
  \begin{pgfonlayer}{nodelayer}
    \node [style=greenrot] (0) at (0, 0) {$\alpha$};
    \node [style=none] (1) at (-1, 0.75) {};
    \node [style=none] (2) at (-0.5, 0.75) {$\cdots$};
    \node [style=none] (3) at (0, 0.75) {};
    \node [style=none] (4) at (-1, -0.75) {};
    \node [style=none] (5) at (-0.5, -0.75) {$\cdots$};
    \node [style=none] (6) at (0, -0.75) {};
  \end{pgfonlayer}
  \begin{pgfonlayer}{edgelayer}
    \draw (1.center) to (0);
    \draw (3.center) to (0);
    \draw (4.center) to (0);
    \draw (0) to (6.center);
  \end{pgfonlayer}
\end{tikzpicture}
}}

\newcommand{\RULEidentitylhs}{\smdiag{%
\begin{tikzpicture}
  \begin{pgfonlayer}{nodelayer}
    \node [style=greenrot] (0) at (0, 0) {};
    \node [style=none] (1) at (0, 0.75) {};
    \node [style=none] (2) at (0, -0.75) {};
  \end{pgfonlayer}
  \begin{pgfonlayer}{edgelayer}
    \draw (1.center) to (0);
    \draw (0) to (2.center);
  \end{pgfonlayer}
\end{tikzpicture}
}}
\newcommand{\RULEidentityrhs}{\smdiag{%
\begin{tikzpicture}
  \begin{pgfonlayer}{nodelayer}
    \node [style=none] (0) at (0, 0.75) {};
    \node [style=none] (1) at (0, -0.75) {};
  \end{pgfonlayer}
  \begin{pgfonlayer}{edgelayer}
    \draw (0.center) to (1.center);
  \end{pgfonlayer}
\end{tikzpicture}
}}

\newcommand{\RULEpicommutelhs}{\smdiag{%
    \begin{tikzpicture}
  \begin{pgfonlayer}{nodelayer}
    \node [style=none] (0) at (0, 0.75) {};
    \node [style=none] (1) at (0, -1) {$\cdots$};
    \node [style=none] (2) at (-0.5, -1) {};
    \node [style=none] (3) at (0.5, -1) {};
    \node [style=redrot] (4) at (0, 0.25) {$\pi$};
    \node [style=greenrot] (5) at (0, -0.5) {$\alpha$};
  \end{pgfonlayer}
  \begin{pgfonlayer}{edgelayer}
    \draw (5) to (2.center);
    \draw (5) to (3.center);
    \draw (5) to (4);
    \draw (4) to (0.center);
  \end{pgfonlayer}
\end{tikzpicture}
}}
\newcommand{\RULEpicommuterhs}{\smdiag{%
\begin{tikzpicture}
  \begin{pgfonlayer}{nodelayer}
    \node [style=none] (0) at (0, 1) {};
    \node [style=none] (1) at (0, -1) {$\cdots$};
    \node [style=none] (2) at (-0.5, -1) {};
    \node [style=none] (3) at (0.5, -1) {};
    \node [style=redrot] (4) at (-0.25, -0.5) {$\pi$};
    \node [style=greenrot] (5) at (0, 0.25) {$-\alpha$};
    \node [style=redrot] (6) at (0.25, -0.5) {$\pi$};
  \end{pgfonlayer}
  \begin{pgfonlayer}{edgelayer}
    \draw (5) to (4);
    \draw (0.center) to (5);
    \draw (5) to (6);
    \draw (4) to (2.center);
    \draw (6) to (3.center);
  \end{pgfonlayer}
\end{tikzpicture}
}}

\newcommand{\RULEcopyinglhs}{\smdiag{%
    \begin{tikzpicture}
  \begin{pgfonlayer}{nodelayer}
    \node [style=none] (1) at (0, -1) {$\cdots$};
    \node [style=none] (2) at (-0.5, -1) {};
    \node [style=none] (3) at (0.5, -1) {};
    \node [style=reddot] (4) at (0, 0.25) {};
    \node [style=greenrot] (5) at (0, -0.5) {$\alpha$};
  \end{pgfonlayer}
  \begin{pgfonlayer}{edgelayer}
    \draw (5) to (2.center);
    \draw (5) to (3.center);
    \draw (5) to (4);
  \end{pgfonlayer}
\end{tikzpicture}
}}
\newcommand{\RULEcopyingrhs}{\smdiag{%
\begin{tikzpicture}
  \begin{pgfonlayer}{nodelayer}
    \node [style=none] (1) at (0, -1) {$\cdots$};
    \node [style=none] (2) at (-0.5, -1) {};
    \node [style=none] (3) at (0.5, -1) {};
    \node [style=reddot] (4) at (-0.5, -0.5) {};
    \node [style=reddot] (6) at (0.5, -0.5) {};
  \end{pgfonlayer}
  \begin{pgfonlayer}{edgelayer}
    \draw (4) to (2.center);
    \draw (6) to (3.center);
  \end{pgfonlayer}
\end{tikzpicture}
}}

\newcommand{\RULEbialgebralhs}{\smdiag{%
\begin{tikzpicture}
  \begin{pgfonlayer}{nodelayer}
    \node [style=none] (0) at (-1, 0.75) {};
    \node [style=none] (1) at (0, 0.75) {};
    \node [style=none] (2) at (-1, -1.75) {};
    \node [style=none] (3) at (0, -1.75) {};
    \node [style=greendot] (4) at (-1, 0) {};
    \node [style=greendot] (5) at (0, 0) {};
    \node [style=reddot] (6) at (-1, -1) {};
    \node [style=reddot] (7) at (0, -1) {};
  \end{pgfonlayer}
  \begin{pgfonlayer}{edgelayer}
    \draw (0.center) to (4);
    \draw (4) to (7);
    \draw (7) to (3.center);
    \draw (2.center) to (6);
    \draw (6) to (5);
    \draw (5) to (1.center);
    \draw (4) to (6);
    \draw (5) to (7);
  \end{pgfonlayer}
\end{tikzpicture}
}}
\newcommand{\RULEbialgebrarhs}{\smdiag{%
\begin{tikzpicture}
  \begin{pgfonlayer}{nodelayer}
    \node [style=none] (0) at (-1, 0.75) {};
    \node [style=none] (1) at (0, 0.75) {};
    \node [style=none] (2) at (-1, -1.75) {};
    \node [style=none] (3) at (0, -1.75) {};
    \node [style=reddot] (4) at (-0.5, 0) {};
    \node [style=greendot] (5) at (-0.5, -1) {};
  \end{pgfonlayer}
  \begin{pgfonlayer}{edgelayer}
    \draw (1.center) to (4);
    \draw (4) to (5);
    \draw (5) to (2.center);
    \draw (3.center) to (5);
    \draw (4) to (0.center);
  \end{pgfonlayer}
\end{tikzpicture}
}}

\newcommand{\RULEhopflhs}{\smdiag{%
\begin{tikzpicture}
  \begin{pgfonlayer}{nodelayer}
    \node [style=none] (0) at (-1, 0.75) {};
    \node [style=none] (1) at (0, 0.75) {};
    \node [style=none] (2) at (-1, -1.75) {};
    \node [style=none] (3) at (0, -1.75) {};
    \node [style=redrot] (4) at (-0.5, 0) {$\alpha$};
    \node [style=greenrot] (5) at (-0.5, -1) {$\beta$};
    \node [style=none] (6) at (-0.5, 0.75) {$\cdots$};
    \node [style=none] (7) at (-0.5, -1.75) {$\cdots$};
  \end{pgfonlayer}
  \begin{pgfonlayer}{edgelayer}
    \draw (1.center) to (4);
    \draw [bend right=45, looseness=1.25] (4) to (5);
    \draw (5) to (2.center);
    \draw (3.center) to (5);
    \draw (4) to (0.center);
    \draw [bend left=45, looseness=1.25] (4) to (5);
  \end{pgfonlayer}
\end{tikzpicture}
}}
\newcommand{\RULEhopfrhs}{\smdiag{%
\begin{tikzpicture}
  \begin{pgfonlayer}{nodelayer}
    \node [style=none] (0) at (-1, 0.75) {};
    \node [style=none] (1) at (0, 0.75) {};
    \node [style=none] (2) at (-1, -1.75) {};
    \node [style=none] (3) at (0, -1.75) {};
    \node [style=redrot] (4) at (-0.5, 0) {$\alpha$};
    \node [style=greenrot] (5) at (-0.5, -1) {$\beta$};
    \node [style=none] (6) at (-0.5, 0.75) {$\cdots$};
    \node [style=none] (7) at (-0.5, -1.75) {$\cdots$};
  \end{pgfonlayer}
  \begin{pgfonlayer}{edgelayer}
    \draw (1.center) to (4);
    \draw (5) to (2.center);
    \draw (3.center) to (5);
    \draw (4) to (0.center);
  \end{pgfonlayer}
\end{tikzpicture}
}}

\newcommand{\RULEcolorchangelhs}{\smdiag{%
\begin{tikzpicture}
  \begin{pgfonlayer}{nodelayer}
    \node [style=none] (0) at (0, 0.75) {};
    \node [style=none] (1) at (0, -1) {$\cdots$};
    \node [style=none] (2) at (-0.5, -1) {};
    \node [style=none] (3) at (0.5, -1) {};
    \node [style=box] (4) at (0, 0.25) {};
    \node [style=greenrot] (5) at (0, -0.5) {$\alpha$};
  \end{pgfonlayer}
  \begin{pgfonlayer}{edgelayer}
    \draw (5) to (2.center);
    \draw (5) to (3.center);
    \draw (5) to (4);
    \draw (4) to (0.center);
  \end{pgfonlayer}
\end{tikzpicture}
}}
\newcommand{\RULEcolorchangerhs}{\smdiag{%
\begin{tikzpicture}
  \begin{pgfonlayer}{nodelayer}
    \node [style=none] (0) at (0, 0.75) {};
    \node [style=none] (1) at (0, -1) {$\cdots$};
    \node [style=none] (2) at (-0.5, -1) {};
    \node [style=none] (3) at (0.5, -1) {};
    \node [style=redrot] (4) at (0, 0.25) {$\alpha$};
    \node [style=box] (5) at (-0.25, -0.5) {};
    \node [style=box] (6) at (0.25, -0.5) {};
  \end{pgfonlayer}
  \begin{pgfonlayer}{edgelayer}
    \draw (5) to (2.center);
    \draw (5) to (4);
    \draw (4) to (0.center);
    \draw (4) to (6);
    \draw (6) to (3.center);
  \end{pgfonlayer}
\end{tikzpicture}
}}

 \newcommand{\RULEhsqrdlhs}{\smdiag{%
\begin{tikzpicture}
  \begin{pgfonlayer}{nodelayer}
    \node [style=none] (0) at (0, 0.25) {};
    \node [style=none] (1) at (0, -1.25) {};
    \node [style=box] (2) at (0, -0.25) {};
    \node [style=box] (3) at (0, -0.75) {};
  \end{pgfonlayer}
  \begin{pgfonlayer}{edgelayer}
    \draw (3) to (1.center);
    \draw (3) to (2);
    \draw (2) to (0.center);
  \end{pgfonlayer}
\end{tikzpicture}
 }}

\newcommand{\RULEhsqrdrhs}{\smdiag{%
\begin{tikzpicture}
  \begin{pgfonlayer}{nodelayer}
    \node [style=none] (0) at (0, 0.25) {};
    \node [style=none] (1) at (0, -1.25) {};
  \end{pgfonlayer}
  \begin{pgfonlayer}{edgelayer}
    \draw (0.center) to (1.center);
  \end{pgfonlayer}
\end{tikzpicture}
}}

\newcommand{\RULEpicopyinglhs}{\smdiag{%
    \begin{tikzpicture}
  \begin{pgfonlayer}{nodelayer}
    \node [style=none] (1) at (0, -1) {$\cdots$};
    \node [style=none] (2) at (-0.5, -1) {};
    \node [style=none] (3) at (0.5, -1) {};
    \node [style=redrot] (4) at (0, 0.25) {$\pi$};
    \node [style=greenrot] (5) at (0, -0.5) {};
  \end{pgfonlayer}
  \begin{pgfonlayer}{edgelayer}
    \draw (5) to (2.center);
    \draw (5) to (3.center);
    \draw (5) to (4);
  \end{pgfonlayer}
\end{tikzpicture}
}}
\newcommand{\RULEpicopyingrhs}{\smdiag{%
\begin{tikzpicture}
  \begin{pgfonlayer}{nodelayer}
    \node [style=none] (1) at (0, -1) {$\cdots$};
    \node [style=none] (2) at (-0.5, -1) {};
    \node [style=none] (3) at (0.5, -1) {};
    \node [style=redrot] (4) at (-0.5, -0.5) {$\pi$};
    \node [style=redrot] (6) at (0.5, -0.5) {$\pi$};
  \end{pgfonlayer}
  \begin{pgfonlayer}{edgelayer}
    \draw (4) to (2.center);
    \draw (6) to (3.center);
  \end{pgfonlayer}
\end{tikzpicture}
}}

\newcommand{\RULEpiBIScopyinglhs}{\smdiag{%
    \begin{tikzpicture}
  \begin{pgfonlayer}{nodelayer}
    \node [style=none] (1) at (0, -1) {$\cdots$};
    \node [style=none] (2) at (-0.5, -1) {};
    \node [style=none] (3) at (0.5, -1) {};
    \node [style=redrot] (4) at (0, 0.25) {};
    \node [style=greenrot] (5) at (0, -0.5) {};
  \end{pgfonlayer}
  \begin{pgfonlayer}{edgelayer}
    \draw (5) to (2.center);
    \draw (5) to (3.center);
    \draw (5) to (4);
  \end{pgfonlayer}
\end{tikzpicture}
}}
\newcommand{\RULEpiBIScopyingrhs}{\smdiag{%
\begin{tikzpicture}
  \begin{pgfonlayer}{nodelayer}
    \node [style=none] (1) at (0, -1) {$\cdots$};
    \node [style=none] (2) at (-0.5, -1) {};
    \node [style=none] (3) at (0.5, -1) {};
    \node [style=redrot] (4) at (-0.5, -0.5) {};
    \node [style=redrot] (6) at (0.5, -0.5) {};
  \end{pgfonlayer}
  \begin{pgfonlayer}{edgelayer}
    \draw (4) to (2.center);
    \draw (6) to (3.center);
  \end{pgfonlayer}
\end{tikzpicture}
}}

\newcommand{\EQNweakpicommutelhs}{\smdiag{%
\begin{tikzpicture}
  \begin{pgfonlayer}{nodelayer}
    \node [style=none] (0) at (0, 0.25) {};
    \node [style=none] (1) at (0, -1.25) {};
    \node [style=greenrot] (2) at (0, -0.25) {$\pi$};
    \node [style=redrot] (3) at (0, -0.75) {$\pi$};
  \end{pgfonlayer}
  \begin{pgfonlayer}{edgelayer}
    \draw (3) to (1.center);
    \draw (3) to (2);
    \draw (2) to (0.center);
  \end{pgfonlayer}
\end{tikzpicture}
}}

\newcommand{\PFweakpicommuterhsi}{\smdiag{%
\begin{tikzpicture}
  \begin{pgfonlayer}{nodelayer}
    \node [style=none] (0) at (0, 0.25) {};
    \node [style=none] (1) at (0, -1.25) {};
    \node [style=greenrot] (2) at (0, -0.25) {$\pi$};
    \node [style=redrot] (3) at (0, -0.75) {};
    \node [style=redrot] (4) at (-0.75, -1.25) {$\pi$};
  \end{pgfonlayer}
  \begin{pgfonlayer}{edgelayer}
    \draw (3) to (1.center);
    \draw (3) to (2);
    \draw (2) to (0.center);
    \draw (3) to (4);
  \end{pgfonlayer}
\end{tikzpicture}
}}

\newcommand{\PFweakpicommuterhsii}{\smdiag{%
\begin{tikzpicture}
  \begin{pgfonlayer}{nodelayer}
    \node [style=none] (0) at (-0.25, 0.25) {};
    \node [style=none] (1) at (0, -1.75) {};
    \node [style=reddot] (2) at (-0.25, -0.25) {};
    \node [style=redrot] (3) at (-0.75, -1.5) {$\pi$};
    \node [style=greenrot] (4) at (-0.75, -0.75) {$\pi$};
    \node [style=greenrot] (5) at (0, -1) {$\pi$};
  \end{pgfonlayer}
  \begin{pgfonlayer}{edgelayer}
    \draw (2) to (0.center);
    \draw (2) to (4);
    \draw (2) to (5);
    \draw (4) to (3);
    \draw (5) to (1.center);
  \end{pgfonlayer}
\end{tikzpicture}
}}

\newcommand{\PFweakpicommuterhsiii}{\smdiag{%
\begin{tikzpicture}
  \begin{pgfonlayer}{nodelayer}
    \node [style=none] (0) at (-0.25, 0.25) {};
    \node [style=none] (1) at (0, -1.75) {};
    \node [style=reddot] (2) at (-0.25, -0.25) {};
    \node [style=redrot] (3) at (-0.75, -1.5) {$\pi$};
    \node [style=greenrot] (4) at (-0.75, -0.75) {};
    \node [style=greenrot] (5) at (0, -1) {$\pi$};
    \node [style=greenrot] (6) at (-1.25, -0.25) {$\pi$};
  \end{pgfonlayer}
  \begin{pgfonlayer}{edgelayer}
    \draw (2) to (0.center);
    \draw (2) to (4);
    \draw (2) to (5);
    \draw (4) to (3);
    \draw (5) to (1.center);
    \draw (4) to (6);
  \end{pgfonlayer}
\end{tikzpicture}
}}

\newcommand{\PFweakpicommuterhsiv}{\smdiag{%
\begin{tikzpicture}
  \begin{pgfonlayer}{nodelayer}
    \node [style=none] (0) at (-0.25, 0.25) {};
    \node [style=none] (1) at (0, -1.75) {};
    \node [style=reddot] (2) at (-0.25, -0.25) {};
    \node [style=redrot] (3) at (-1.5, -1) {$\pi$};
    \node [style=greenrot] (4) at (0, -1) {$\pi$};
    \node [style=greenrot] (5) at (-1.5, -0.25) {$\pi$};
    \node [style=redrot] (6) at (-0.75, -1) {$\pi$};
  \end{pgfonlayer}
  \begin{pgfonlayer}{edgelayer}
    \draw (2) to (0.center);
    \draw (2) to (4);
    \draw (4) to (1.center);
    \draw (5) to (3);
    \draw (6) to (2);
  \end{pgfonlayer}
\end{tikzpicture}
}}

\newcommand{\EQNweakpicommuterhs}{\smdiag{%
\begin{tikzpicture}
  \begin{pgfonlayer}{nodelayer}
    \node [style=none] (0) at (0, 0.25) {};
    \node [style=none] (1) at (0, -1.25) {};
    \node [style=redrot] (2) at (0, -0.25) {$\pi$};
    \node [style=greenrot] (3) at (0, -0.75) {$\pi$};
  \end{pgfonlayer}
  \begin{pgfonlayer}{edgelayer}
    \draw (3) to (1.center);
    \draw (3) to (2);
    \draw (2) to (0.center);
  \end{pgfonlayer}
\end{tikzpicture}
}}

\newcommand{\EQNweakpihomomorphismlhs}{\smdiag{%
    \begin{tikzpicture}
  \begin{pgfonlayer}{nodelayer}
    \node [style=none] (0) at (0, 0.75) {};
    \node [style=none] (2) at (-0.5, -1) {};
    \node [style=none] (3) at (0.5, -1) {};
    \node [style=redrot] (4) at (0, 0.25) {$\pi$};
    \node [style=greendot] (5) at (0, -0.5) {};
  \end{pgfonlayer}
  \begin{pgfonlayer}{edgelayer}
    \draw (5) to (2.center);
    \draw (5) to (3.center);
    \draw (5) to (4);
    \draw (4) to (0.center);
  \end{pgfonlayer}
\end{tikzpicture}
}}

\newcommand{\PFEweakpihomomorphismi}{\smdiag{%
\begin{tikzpicture}
  \begin{pgfonlayer}{nodelayer}
    \node [style=none] (0) at (0.5, 0.75) {};
    \node [style=none] (2) at (-0.5, -1) {};
    \node [style=none] (3) at (0.5, -1) {};
    \node [style=redrot] (4) at (0, 0.25) {};
    \node [style=greendot] (5) at (0, -0.5) {};
    \node [style=redrot] (6) at (-0.5, 0.75) {$\pi$};
  \end{pgfonlayer}
  \begin{pgfonlayer}{edgelayer}
    \draw (5) to (2.center);
    \draw (5) to (3.center);
    \draw (5) to (4);
    \draw (4) to (0.center);
    \draw (4) to (6);
  \end{pgfonlayer}
\end{tikzpicture}
}}

\newcommand{\PFEweakpihomomorphismii}{\smdiag{%
\begin{tikzpicture}
  \begin{pgfonlayer}{nodelayer}
    \node [style=none] (0) at (0.25, 1) {};
    \node [style=none] (1) at (-0.75, -1.25) {};
    \node [style=none] (2) at (0.5, -1.25) {};
    \node [style=greendot] (3) at (0.25, 0.25) {};
    \node [style=reddot] (4) at (0.25, -0.75) {};
    \node [style=redrot] (5) at (-1, 0.75) {$\pi$};
    \node [style=reddot] (6) at (-0.5, -0.75) {};
    \node [style=greendot] (7) at (-0.5, 0.25) {};
  \end{pgfonlayer}
  \begin{pgfonlayer}{edgelayer}
    \draw (4) to (2.center);
    \draw (4) to (3);
    \draw (3) to (0.center);
    \draw (1.center) to (6);
    \draw (6) to (7);
    \draw (6) to (3);
    \draw (7) to (4);
    \draw (7) to (5);
  \end{pgfonlayer}
\end{tikzpicture}
}}
\newcommand{\PFEweakpihomomorphismiii}{\smdiag{%
\begin{tikzpicture}
  \begin{pgfonlayer}{nodelayer}
    \node [style=none] (0) at (0.25, 1) {};
    \node [style=none] (1) at (-0.75, -1.25) {};
    \node [style=none] (2) at (0.5, -1.25) {};
    \node [style=greendot] (3) at (0.25, 0.25) {};
    \node [style=reddot] (4) at (0.25, -0.75) {};
    \node [style=redrot] (5) at (-1, -0.25) {$\pi$};
    \node [style=reddot] (6) at (-0.5, -0.75) {};
    \node [style=redrot] (7) at (-0.5, 0.5) {$\pi$};
  \end{pgfonlayer}
  \begin{pgfonlayer}{edgelayer}
    \draw (4) to (2.center);
    \draw (4) to (3);
    \draw (3) to (0.center);
    \draw (1.center) to (6);
    \draw (6) to (3);
    \draw (5) to (6);
    \draw (7) to (4);
  \end{pgfonlayer}
\end{tikzpicture}
}}
\newcommand{\EQNweakpihomomorphismrhs}{\smdiag{%
\begin{tikzpicture}
  \begin{pgfonlayer}{nodelayer}
    \node [style=none] (0) at (0, 1) {};
    \node [style=none] (2) at (-0.5, -1) {};
    \node [style=none] (3) at (0.5, -1) {};
    \node [style=redrot] (4) at (-0.25, -0.5) {$\pi$};
    \node [style=greenrot] (5) at (0, 0.25) {};
    \node [style=redrot] (6) at (0.25, -0.5) {$\pi$};
  \end{pgfonlayer}
  \begin{pgfonlayer}{edgelayer}
    \draw (5) to (4);
    \draw (0.center) to (5);
    \draw (5) to (6);
    \draw (4) to (2.center);
    \draw (6) to (3.center);
  \end{pgfonlayer}
\end{tikzpicture}
}}

\title{Pivoting makes the \zxcalculus complete for real stabilizers}
\author{Ross Duncan
\institute{University of Strathclyde\\ Glasgow, UK}
\email{ross.duncan@strath.ac.uk}
\and
Simon Perdrix
\institute{CNRS, LORIA UMR 7503, CARTE Project-Team\\ Nancy, France}
\email{simon.perdrix@loria.fr}
}

\def\titlerunning{Pivoting makes the \zxcalculus complete for real stabilizers}
\def\authorrunning{R. Duncan \& S. Perdrix}

\begin{document}

\maketitle

\begin{abstract}
  We show that pivoting property of graph states cannot be derived
  from the axioms of the \zxcalculus, and that pivoting does not imply
  local complementation of graph states.  Therefore the \zxcalculus
  augmented with pivoting is strictly weaker than the calculus
  augmented with the Euler decomposition of the Hadamard gate.  We
  derive an angle-free version of the  \zxcalculus and show that it is
  complete for real stabilizer quantum mechanics.
\end{abstract}

The \zxcalculus is a formal theory for reasoning about quantum
computational systems \cite{Coecke:2009aa}.  It consists of a
graphical language based on the Pauli $Z$ and $X$ observables, and a
collection of axioms expressed as graph rewrite rules.  The
\zxcalculus is expressive enough to represent any quantum circuit, and
its equations are complete for the stabilizer fragment of quantum
mechanics \cite{Backens:2012fk}.  Due to its graphical nature, and its
close relationship to the $Z$ and $X$ observables, the
\zxcalculus is particularly well adapted to the study of graph states
and measurement-based quantum computation
\cite{Duncan:2010aa,Duncan:2012uq}.

In addition to the two observables, the \zxcalculus also contains an
operator for the Hadamard map: this is the map which exchanges the $Z$
and $X$ bases, and thus provides a duality principle for the graphical
language.  In previous work \cite{Duncan:2009ph} the authors showed
that if the Hadamard can be expressed in terms of $Z$ and $X$
rotations---that is, as an Euler decomposition---then Van Den Nest's
theorem \cite{VdN04} about local complementation of graph states
follows, and vice versa.  Furthermore, these results cannot be derived
from the original axioms, hence the theory ``\zxcalculus + Euler'' is
strictly stronger than the plain \zxcalculus.

In this paper we find a theory intermediate between the two, albeit
having a similar flavour.  We consider an operation on graph states
called \emph{pivoting} and show that its defining property is
equivalent to the possibility to express (one of) the Pauli matrices
in terms of the Hadamard.  Since pivoting can be done via local
complementation, ``\zxcalculus + Euler'' is stronger than
``\zxcalculus + Pivot''.  However, we will show that, once again,
these equations cannot be derived from the plain \zxcalculus.

The theory ``\zxcalculus + Euler'' is known to be complete for the
stabiliser fragment of quantum mechanics \cite{Backens:2012fk}: the
stabiliser fragment corresponds to the sub-calculus where all angles
are multiples of $\pi/2$. We show that the intermediate calculus
``\zxcalculus + Pivot'' is complete for the \emph{real} stabiliser
fragment of quantum mechanics, and that this fragment admits an
angle-free axiomatisation.

Real quantum mechanics is sufficient for quantum computing
\cite{BV97}, in the sense that any unitary evolution on $n$-qubits
can be simulated (using a simple encoding) by a real unitary evolution
acting on $n+1$ qubits. As a consequence, while not complete for
(complex) quantum mechanics, the intermediate calculus ``\zxcalculus +
Pivot'' might be useful and simpler than the the whole ``\zxcalculus +
Euler'' for proving properties of quantum systems, for example via
rewriting.  


\begin{remark}\label{rem:what-are-the-axioms}
  There is some variation about which axioms comprise the \zxcalculus.
  In \cite{Duncan:2009ph} we considered fewer axioms than we do here;
  whereas in some later work, notably \cite{Backens:2012fk}, the Euler
  decomposition of the Hadamard is included as an axiom.
\end{remark}


\section{The Graphical Formalism}
\label{sec:graphical-formalism}

We recall  the syntax, semantics, and basic properties of the
\zxcalculus.  For a full exposition, see \cite{Coecke:2009aa}.

\begin{definition}
  \label{def:opengraph}
  An \emph{open graph} is a triple $(G,I,O)$ consisting of a finite
  undirected graph $G =(V,E)$ and distinguished subsets $I,O \subseteq
  V$ of degree one vertices, called the \emph{inputs} and
  \emph{outputs}, respectively.  The set of vertices $I\cup O$ is
  called the \emph{boundary} of $G$, and $V\setminus (I\cup O)$ is the
  \emph{interior} of $G$.  An open graph is called \emph{empty} if its
  interior is empty; it is called \emph{prime} if it is connected and
  its interior is a singleton.
\end{definition}

We view the inputs and outputs as finite ordinals, and write
$\gamma:n\to m$ for a graph with $n$ inputs and $m$ outputs.  Open
graphs form a self-dual compact category: composition is achieved by
identifying the inputs of one graph with the outputs or another and
erasing the resulting vertices; the tensor product is simple
juxtaposition of graphs.  The unit and counit maps are generated from
the unique empty graphs $d:0\to 2$ and $e:2\to 0$.  Note that, due to
general results
\cite{JS:1991:GeoTenCal1,Selinger:dagger:2005,Duncan:thesis:2006}, a
pair of graphs can be deformed from one to other if and only if they
are equal by the axioms of compact categories.

The terms of the \zxcalculus are certain open graphs we call
\emph{diagrams}.

\begin{definition}\label{def:diagrams}
  A \emph{diagram} is an arrow of the free category \catD
  generated by the following prime graphs:
  \[
  Z^n_m(\alpha) = \greenspider{\alpha} 
  \qquad 
  X^n_m(\alpha) = \redspider{\alpha} 
  \qquad 
  H = \smdiag{\begin{tikzpicture}
	\begin{pgfonlayer}{nodelayer}
		\node [style=empty] (0) at (-1, 2.25) {};
		\node [style=empty] (1) at (-1, 1.25) {};
		\node [style=box] (2) at (-1, 1.75) {};
	\end{pgfonlayer}
	\begin{pgfonlayer}{edgelayer}
		\draw (0) to (2);
		\draw (2) to (1);
	\end{pgfonlayer}
\end{tikzpicture}}
  \]
  where $n$ and $m$ are the number of inputs and outputs respectively,
  and $\alpha \in [0,2\pi)$ is  an angle called \emph{phase}.  If
  $\alpha =  0 $ it will be omitted from the diagram.
\end{definition}

We define the semantics of diagrams via an interpretation functor 
$\denote{\cdot}: \catD \to \wfdhilb$,
where $\wfdhilb$ is  the
category of complex Hilbert spaces and linear maps under the equivalence
relation $f \equiv g$ iff there exists $\theta$ such that $f =
e^{i\theta} g$.  A diagram $f:n\to m$ output defines a
linear map  $\denote{f} : \mathbb{C}^{\otimes 2n} \to
\mathbb{C}^{\otimes 2m}$ as follows:
\begin{align*}
  \denote{Z^n_m(\alpha)} &= \left\{
  \begin{array}{rcl}
    \ket{0}^n & \mapsto & \ket{0}^m \\
    \ket{1}^n & \mapsto & e^{i\alpha}\ket{1}^m 
  \end{array} \right.  
\\
\\
\denote{X^n_m(\beta)} &= \left\{
\begin{array}{rcl}
  \ket{+}^n & \mapsto & \ket{+}^m \\
  \ket{-}^n & \mapsto & e^{i\beta}\ket{-}^m 
\end{array} \right.
\\
\\
  \denote{H} &= \frac{1}{\sqrt{2}}
  \begin{pmatrix*}[r]
    1&1\\1&-1
  \end{pmatrix*} \,.
\end{align*}
The map $\denote{\cdot}$ extends in the evident way to a monoidal
functor.  We can now see where the name \zxcalculus calculus comes
from:  the $Z$ vertices are defined in terms of the $Z$ basis
of $\mathbb{C}^2$ while the $X$ vertices are defined in terms of the
$X$ basis.

The interpretation of \catD contains a universal set of quantum gates.
Note that $Z^1_1(\alpha)$ and $X_1^1(\alpha)$ are the rotations around the
$X$ and $Z$ axes, and in particular when $\alpha = \pi$ they yield
the Pauli $X$ and $Z$ matrices. The \CZ is defined by:
\[
\CZ ~=  ~\vcenter{\hbox{\input{diag/contZ.tikz}}}
\]

In order to obtain the \zxcalculus we quotient the free category \catD
by the equations shown in Figure~\ref{fig:eqns}; the quotient category
we denote by $\mathbb{D}$.

\begin{figure}[t]
  \[
\begin{array}{rclcrclcrcl}
\RULEgreenspiderlhs &=& \RULEgreenspiderrhs 
&&
\RULEantilooplhs &=& \RULEantilooprhs
&&
\RULEidentitylhs &=& \RULEidentityrhs
\\
& \text{(S1)} &&&& \text{(S2)} &&&& \text{(S3)}
\\
\\
\RULEpicommutelhs &=& \RULEpicommuterhs
&&
\RULEcopyinglhs &=& \RULEcopyingrhs
&&
\RULEcolorchangelhs &=& \RULEcolorchangerhs 
\\
& \text{($\pi$)} &&&& \text{(C)} &&&& \text{(H1)}
\\
\\
\RULEhopflhs &=& \RULEhopfrhs
&&
\RULEbialgebralhs &=& \RULEbialgebrarhs 
&&
\RULEhsqrdlhs &=& \RULEhsqrdrhs
\\
& \text{(Hpf)} &&&& \text{(Bi)} &&&& \text{(H2)}
\end{array}
\]
  \caption{Equations for the \zxcalculus}
  \label{fig:eqns}
\end{figure}

The equations of Fig~\ref{fig:eqns} are sound with respect to the interpretation
functor $\denote{\cdot}$ introduced above.

\begin{proposition}  \label{prop:sound}
  There exists a canonical functor $\denote{\cdot}_\sim: \mathbb{D} \to
  \wfdhilb$ making the following diagram commute:
  \begin{diagram}
    \catD & \rOnto & \mathbb{D} \\
    & \rdTo<{\denote{\cdot}} & \dTo>{\denote{\cdot}_\sim} \\
    && \wfdhilb
  \end{diagram}
\end{proposition}

\noindent
In the rest of this paper we won't make any distinction between
$\catD$ and $\mathbb{D}$, nor between the interpretation functors.
Indeed, we will abuse notation and refer to both as $\denote{\cdot}$.

\begin{remark}\label{rem:diff-generators}
  Note that in the presence of the equations (S1)-(S3), which we refer
  to collectively as the ``spider rule'', we could have made other
  choices for the generators of the $X$ and $Z$ families of vertices.
  For example, the prime graphs
  \[
  \delta = \vcenter{\hbox{\input{diag/greendelta.tikz}}}\qquad\quad
  \epsilon = \vcenter{\hbox{\begin{tikzpicture}
	\begin{pgfonlayer}{nodelayer}
		\node [style=empty] (0) at (-1, 2) {};
		\node [style=greenrot] (1) at (-1, 1.5) {};
	\end{pgfonlayer}
	\begin{pgfonlayer}{edgelayer}
		\draw (0) to (1);
	\end{pgfonlayer}
\end{tikzpicture}}} \qquad\quad
  p_\alpha = \vcenter{\hbox{\begin{tikzpicture}
	\begin{pgfonlayer}{nodelayer}
		\node [style=empty] (0) at (-0.5, 2.5) {};
		\node [style=empty] (1) at (-0.5, 1.5) {};
		\node [style=greendot] (2) at (-0.5, 2) {$\alpha$};
	\end{pgfonlayer}
	\begin{pgfonlayer}{edgelayer}
		\draw (2) to (1);
		\draw (2) to (0);
	\end{pgfonlayer}
\end{tikzpicture}}}
  \]
  are used in the formulation that emphasises the fact that each
  family forms a Frobenius algebra.  From that perspective the spider
  rule is effectively a normal-form theorem; see \cite{Coecke:2009db}
  for details.
\end{remark}

\begin{proposition}
  \label{prop:basicprops} The following are direct consequences of the axioms.
  \begin{itemize}
  \item Any connected diagram containing only $Z$ or only $X$ vertices
    is equivalent to a prime graph.
  \item Any diagram without any $H$ is equivalent to a simple bipartite graph.
  \item Any diagram is equivalent to (a) a diagram with no $Z$ vertices;
    and (b) a diagram with no $X$ vertices.
  \item Any equation which holds between two graphs, also holds
    with $Z$ and $X$ exchanged.
  \end{itemize}
\end{proposition}

\begin{remark}
  Note that although Figure \ref{fig:eqns} seems to favour one colour
  over the other, by the last point of Proposition
  \ref{prop:basicprops} we know that all the rules apply with the
  colours reversed.\label{rem:2}
\end{remark}

\noindent {\bf Euler decomposition of $H$.} The following axiom is not part of the definition of the \zxcalculus
\begin{equation}
\vcenter{\hbox{\begin{tikzpicture}
	\begin{pgfonlayer}{nodelayer}
		\node [style=empty] (0) at (-1, 2.25) {};
		\node [style=empty] (1) at (-1, 1.25) {};
		\node [style=box] (2) at (-1, 1.75) {};
	\end{pgfonlayer}
	\begin{pgfonlayer}{edgelayer}
		\draw (0) to (2);
		\draw (2) to (1);
	\end{pgfonlayer}
\end{tikzpicture}}} 
= ~
\vcenter{\hbox{\input{diag/axiomEuler-r.tex}}} \label{eq:euler}\tag{EU}
\end{equation}
In \cite{Duncan:2009ph} we proved that the Euler decomposition cannot
be derived from the axioms of the \zxcalculus; however, in that paper  we
considered slightly weaker axioms.  It is straight-forward to give a
counter-model for the \zxcalculus of today.

\begin{lemma}\label{lem:no-euler-in-zx}
  The Euler decomposition of $H$ cannot be derived by the rules of the
  ZX calculus.
 \begin{proof}
   We define an alternative interpretation functor $\denote\cdot
   _0:\catD \to \wfdhilb$ by
   \begin{align*}
     \denote{H}_0 &= \denote{H} \\
     \denote{Z^n_m(\alpha)}_0 &= \denote{Z^n_m(0)} \\
     \denote{X^n_m(\beta)}_0 &= \denote{X^n_m(0)}\,.
   \end{align*}
   It's easy to verify that all the equations of Figure \ref{fig:eqns}
   still hold under $\denote{\cdot}_0$ but \eqref{eq:euler} fails.
 \end{proof}
\end{lemma}

\section{Graph states and Local complementation}
\label{sec:graph-states-local}

\begin{definition}\label{def:graphstate}
  Let $G=(V,E)$ be an undirected graph.  Then the \emph{graph state}
  $\ket G$ is defined by
  \[
  \ket G = \left( \prod_{uv\in E} \CZ_{uv} \right) \bigotimes_{v\in V}
  \ket{+}
  \]
\end{definition}

Given a graph $G$ we can directly write down the diagram $D_G$ such
that $\denote{D_G} = \ket{G}$ as follows: (1) for each $v\in V$ we add
a \emph{Z} vertex, connected to an output; (2) for each edge $uv \in
E$ we add an $H$ vertex, connected to those $Z$ vertices corresponding
to the vertices $u$ and $v$.

\begin{example}\label{ex:1}
  Consider the case when $G$ is just a triangle:
  \[
  G = \smdiag{\begin{tikzpicture}
	\begin{pgfonlayer}{nodelayer}
		\node [style=blackrot] (0) at (0, 1.5) {};
		\node [style=blackrot] (1) at (-0.5, 0.5) {};
		\node [style=blackrot] (2) at (0.5, 0.5) {};
	\end{pgfonlayer}
	\begin{pgfonlayer}{edgelayer}
		\draw (0) to (1);
		\draw (1) to (2);
		\draw (2) to (0);
	\end{pgfonlayer}
\end{tikzpicture}}
  \qquad\qquad
  D_G = \smdiag{\input{diag/triangle-state.tikz}}
  \]
\end{example}

\begin{proposition}\label{prop:stab-graph-state}
  Let $G = (V,E)$ be a graph with $v\in V$ and define
  \[
  K_v = \left(\prod_{u\in N(v)} Z_u \right) X_v \,.
  \]
  Then $K_v\ket G = \ket G$.
  \begin{proof}
    We apply an $X(\pi$) to the output corresponding to $v$ and a
    $Z(\pi)$ on all the outputs of the neighbours of $v$:
    \[
    \smdiag{\input{diag/gs-stab-i.tikz}}
    =
    \smdiag{\input{diag/gs-stab-ii.tikz}}
    =
    \smdiag{\input{diag/gs-stab-iii.tikz}}
    =
    \smdiag{\input{diag/gs-stab-iv.tikz}}
    \]
  \end{proof}
\end{proposition}

\begin{definition}\label{def:local-complementation}
  Suppose $G = (V,E)$ is a graph with $v\in V$.  Let $E_1 = E \cap
  (N(v)\times N(v))$ and $E_2 = (N(v)\times N(v)) \setminus E_1$.
  Then the local complementation of $G$ at $v$ is defined by 
  \[
  G*v =   (V, (E\setminus E_1) \cup E_2) \,.
  \]
  Equivalently, if $u,u'$ are neighbours of $v$ then $uu'$ is an edge
  of $G*v$ if and only if it is not an edge of $G$; otherwise the two
  graphs are the same.
\end{definition}

For graph states local complementation can also be expressed in terms
of a product of single qubit operations:

\begin{proposition}[\cite{VdN04}]\label{thm:vandennest}
  Let $G$ be a graph with vertex $v$; define
  \[
  M_v = \left( \prod_{u\in N(v)} Z(-\pi/2)_u \right) \cdot X(\pi/2)_v
  \]
  Then $\ket{G*v} = M_v\ket{G}$.
\end{proposition}

\noindent
Note that $M_v^2 = K_v$ hence local complementation is involutive on
graph states.

\begin{example}\label{ex:2}
  Here we consider the local complementation of the  triangle by its
  top vertex:
  \[
  \smdiag{\input{diag/triangle-state.tikz}} 
  =
  \smdiag{\input{diag/lc-triangle-state.tikz}} 
  \]
\end{example}

\begin{theorem}[\cite{Duncan:2009ph}]  \label{thm:lc-is-euler}
  Proposition \ref{thm:vandennest} is equivalent to
  Equation~\eqref{eq:euler} in the \zxcalculus, hence it cannot be
  proven in the \zxcalculus.
\end{theorem}

\section{Pivoting}
\label{sec:pivoting}

Pivoting, also known as edge-local complementation, is a local
transformation of graphs. Given a graph $G$ with an edge $uv$,
$G\wedge uv$, the graph obtained by pivoting according to $uv$,
consists in exchanging the two vertices $u$ and $v$ and in
complementing the tripartite subgraph formed by ($i$) the common
neighbours of $u$ and $v$; $(ii)$ the exclusive neighbours of $u$; and
($iii$) the exclusive neighbours of $v$ (see Figure \ref{fg1}).


   \begin{psfrags}
   \psfrag{u}{$u$}
\begin{figure}[h]
   \begin{center}

\includegraphics[width=8cm]{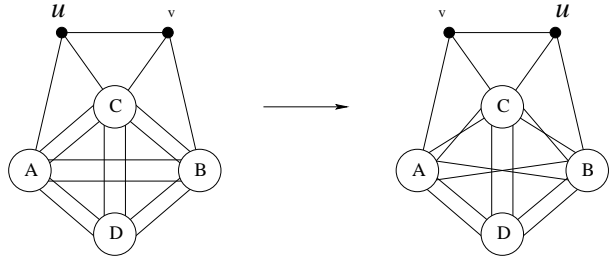}\vspace{-0.5cm}
\end{center}
\caption{\label{fg1}Pivoting on $uv$. $C={ N}(u)\cap{ N}(v)$, $A={ N}(u)\setminus C$, $B={ N}(v)\setminus C$,  and $D$ is the rest of the vertices. Pivoting on $uv$ exchanges vertices $u$ and $v$, and for any $(x,y) \in (A\times B)\cup(B\times C)\cup(A\times C)$, the edge $xy$ is deleted if $xy$ was an edge, and added otherwise.}
\end{figure}    
\end{psfrags}

Pivoting is a combination of local complementations, $G\wedge uv = G * u* v*u$ (Notice that $G*u*v*u=G*v*u*v$) and 
can be performed on graph states by applying Hadamard on vertices $u$ and $v$ and $Z$ on their common neighbours:

\begin{proposition}[\emph{Pivoting Property} \cite{nest2005edge,MP2013}]
\[
\ket {G\wedge uv}  =H_{u,v}Z_{N(u)\cap N(v)} \ket G
\]
\end{proposition}

Pivoting of graph states have several applications in quantum
information processing. In particular the universality of the
triangular grid as a resource of measurement-based quantum computing
has been proved using pivoting \cite{MP2013}; pivoting can also be
used to compute the minimal distance of linear codes 
\cite{EdgeLocal}.

In the rest of this section, we prove that an additional axiom, strictly weaker that the Euler decomposition of $H$, needs to be added to the \zxcalculus to prove  the pivoting property. However, when $u$ and $v$ have no common neighbours, the pivoting property can be proved in the plain \zxcalculus:

\begin{lemma} \label{lem:weak-pivot-in-ZX}
  For any graph $G=(V,E)$ and any $u,v \in V$ which have
  no common neighbour, $\ket {G\wedge uv} = H_{u,v}\ket G$ can be
  derived in the \zxcalculus.
\end{lemma}

\begin{proof}
The proof is based on the  generalised bialgebra law \cite{Duncan:2009ph}:
\[ 
\vcenter{\hbox{\input{diag/k34.tex}}} 
= 
\vcenter{\hbox{\input{diag/e34.tex}}} 
\] 
Assume for the moment that there is no edge between the neighbours of
$u$ and the neighbours of $v$. In that case applying $H$ on both $u$ and
$v$ permutes $u$ and $v$ and creates a complete bipartite graph
between the neighbours of $u$ and this of $v$. For instance, if $u$
and $v$ both have two neighbours:
\[
\vcenter{\hbox{\input{diag/proof_piv_nocom.tikz}}}
=~\vcenter{\hbox{\input{diag/proof_piv_nocom-2.tikz}}}
=~\vcenter{\hbox{\input{diag/proof_piv_nocom-4.tikz}}} 
=~\vcenter{\hbox{\input{diag/proof_piv_nocom-5.tikz}}}  
\]
The first equation is via rule (H1) while the second follows from the
generalised bialgebra.  The final equation uses rule (H1) again to
remove all the red vertices, followed by the spider rule, and some
rearrangement of the graph.
Now suppose that in fact there were some edges between the neighbours
of $u$ and those of $v$.  The procedure above will add an additional edge,
and then both may be removed since
$\vcenter{\hbox{\input{diag/hopf-h-l.tikz}}}=\vcenter{\hbox{\begin{tikzpicture}
	\begin{pgfonlayer}{nodelayer}
		\node [style=greendot] (0) at (-1, 2) {};
		\node [style=empty] (1) at (-1.25, 2) {};
		\node [style=empty] (2) at (-0.25, 2) {};
		\node [style=greendot] (3) at (-0.5, 2) {};
	\end{pgfonlayer}
	\begin{pgfonlayer}{edgelayer}
		\draw (0) to (1);
		\draw (3) to (2);
	\end{pgfonlayer}
\end{tikzpicture}}}$
by (HpF).
\end{proof}

As a consequence, pivoting of triangle-free graphs or bipartite (or
2-colourable) graphs can be derived in the \zxcalculus. Notice that the
pivoting preserves bipartiteness (but not triangle freeness), so one
can prove a series of pivotings on bipartite graphs in the
\zxcalculus.

In the following we prove that the pivoting of arbitrary graph can be
derived in the \zxcalculus augmented with a new rule for
$\pi$-rotations:


\begin{theorem} \label{thm:Hloopsufficient} 
  Pivoting of arbitrary graph can be proved in the \zxcalculus
  augmented with the following axiom:
  \begin{equation}\label{hloop}\tag{HL}
    \vcenter{\hbox{\begin{tikzpicture}
	\begin{pgfonlayer}{nodelayer}
		\node [style=empty] (0) at (-0.5, 2.5) {};
		\node [style=empty] (1) at (-0.5, 1.5) {};
		\node [style=greendot] (2) at (-0.5, 2) {$\pi$};
	\end{pgfonlayer}
	\begin{pgfonlayer}{edgelayer}
		\draw (2) to (1);
		\draw (2) to (0);
	\end{pgfonlayer}
\end{tikzpicture}}} 
    = 
    \vcenter{\hbox{\input{diag/new-rule.tikz}}} 
  \end{equation}
\end{theorem}

\noindent This new axiom is called the \emph{H-loop} axiom as it can be rewritten as $\vcenter{\hbox{\begin{tikzpicture}
	\begin{pgfonlayer}{nodelayer}
		\node [style=empty] (0) at (-0.5, 2.5) {};
		\node [style=empty] (1) at (-0.5, 1.5) {};
		\node [style=greendot] (2) at (-0.5, 2) {$\pi$};
	\end{pgfonlayer}
	\begin{pgfonlayer}{edgelayer}
		\draw (2) to (1);
		\draw (2) to (0);
	\end{pgfonlayer}
\end{tikzpicture}}} 
= 
\vcenter{\hbox{\input{diag/H-loop.tikz}}}$.

\begin{proof} We illustrate the proof on the particular case where the $u$ and $v$ have two common neighbours:
\[ 
\vcenter{\hbox{\input{diag/triangle_pivot-0.tikz}}} 
= 
\vcenter{\hbox{\input{diag/triangle_pivot-4.tikz}}} 
= 
\vcenter{\hbox{\input{diag/triangle_pivot-5.tikz}}} 
= 
\vcenter{\hbox{\input{diag/triangle_pivot-7.tikz}}} = 
\vcenter{\hbox{\input{diag/triangle_pivot-6.tikz}}} 
\] 
The left-most diagram corresponds to a graph state on which $H$ is
applied on two vertices  $u$ and $v$ and $Z$ (green
$\pi$-rotation) on their two common neighbours. The second diagram is
obtained using the (\ref{hloop}) axiom. This transformation splits each
common neighbour of $u$ and $v$ in such a way that
Lemma~\ref{lem:weak-pivot-in-ZX} can be applied, leading to
the third diagram. The application of the spider rule (fourth diagram)
and the Hopf law (fifth diagram) completes the proof for this
particular graph.


The general case is similar.  First, the $\pi$-rotations on the common
neighbours are removed using the (\ref{hloop}) axiom, which splits the common
neighbours. Then, in the absence of common neighbours
Lemma~\ref{lem:weak-pivot-in-ZX}  is
used. Finally spiders and the Hopf Law complete the proof. 
\end{proof}

In the following, we show that one can derive pivoting if and only if Equation~\eqref{hloop} holds:

\begin{lemma}\label{lem:pivot-implies-loop}
In the \zxcalculus, the pivoting property for the triangle implies that the $\pi$-rotation is equivalent to a ``$H$-loop'', i.e.

\[ 
\vcenter{\hbox{\input{diag/tri_pivot-0.tikz}}} 
= 
\vcenter{\hbox{\input{diag/tri_pivot-2.tikz}}} 
~~~~\Rightarrow~~~~~~
\vcenter{\hbox{\begin{tikzpicture}
	\begin{pgfonlayer}{nodelayer}
		\node [style=empty] (0) at (-0.5, 2.5) {};
		\node [style=empty] (1) at (-0.5, 1.5) {};
		\node [style=greendot] (2) at (-0.5, 2) {$\pi$};
	\end{pgfonlayer}
	\begin{pgfonlayer}{edgelayer}
		\draw (2) to (1);
		\draw (2) to (0);
	\end{pgfonlayer}
\end{tikzpicture}}} 
= 
\vcenter{\hbox{\input{diag/H-loop.tikz}}} 
\] \end{lemma}

\begin{proof}$\vcenter{\hbox{\input{diag/H-loop.tikz}}} =\vcenter{\hbox{\input{diag/new-rule.tikz}}}  = \vcenter{\hbox{\input{diag/tri_pivot-3b.tikz}}} = \vcenter{\hbox{\input{diag/tri_pivot-2b.tikz}}} = \vcenter{\hbox{\input{diag/tri_pivot-0b.tikz}}}= \vcenter{\hbox{\input{diag/tri_pivot-4b.tikz}}}= \vcenter{\hbox{\input{diag/tri_pivot-5b.tikz}}}=\vcenter{\hbox{\begin{tikzpicture}
	\begin{pgfonlayer}{nodelayer}
		\node [style=empty] (0) at (-0.5, 2.5) {};
		\node [style=empty] (1) at (-0.5, 1.5) {};
		\node [style=greendot] (2) at (-0.5, 2) {$\pi$};
	\end{pgfonlayer}
	\begin{pgfonlayer}{edgelayer}
		\draw (2) to (1);
		\draw (2) to (0);
	\end{pgfonlayer}
\end{tikzpicture}}} $

%

\end{proof}

\begin{lemma}\label{lem:no-pivot-in-zx}
$\vcenter{\hbox{\begin{tikzpicture}
	\begin{pgfonlayer}{nodelayer}
		\node [style=empty] (0) at (-0.5, 2.5) {};
		\node [style=empty] (1) at (-0.5, 1.5) {};
		\node [style=greendot] (2) at (-0.5, 2) {$\pi$};
	\end{pgfonlayer}
	\begin{pgfonlayer}{edgelayer}
		\draw (2) to (1);
		\draw (2) to (0);
	\end{pgfonlayer}
\end{tikzpicture}}} 
= 
\vcenter{\hbox{\input{diag/H-loop.tikz}}}$ cannot be derived from the rules of the \zxcalculus. 
\end{lemma}

\begin{proof}
 We consider the interpretation functor $\denote ._0$ introduced in
 Lemma~\ref{lem:no-euler-in-zx}, which preserves all the axioms of the
 \zxcalculus, but for which we have:
 \[
 \frac12 
 \begin{pmatrix*}[r]
   1&0\\0&-1
 \end{pmatrix*}
\quad  = \quad 
 \denote{\vcenter{\hbox{\input{diag/H-loop.tikz}}}}_0
\quad  \neq  \quad 
 \denote{\vcenter{\hbox{\begin{tikzpicture}
	\begin{pgfonlayer}{nodelayer}
		\node [style=empty] (0) at (-0.5, 2.5) {};
		\node [style=empty] (1) at (-0.5, 1.5) {};
		\node [style=greendot] (2) at (-0.5, 2) {$\pi$};
	\end{pgfonlayer}
	\begin{pgfonlayer}{edgelayer}
		\draw (2) to (1);
		\draw (2) to (0);
	\end{pgfonlayer}
\end{tikzpicture}}} }_0 
\quad  =  \quad 
 \begin{pmatrix*}[r]
   1&0\\0&1  
 \end{pmatrix*}
 \]
\end{proof}
Like the Euler decomposition of $H$ (Equation \eqref{eq:euler}),
Equation~\eqref{hloop} cannot be derived from the rules of the
\zxcalculus. The completeness for the stabilisers of the \zxcalculus
augmented with Euler decomposition of $H$ guarantees that equation
\eqref{hloop} can be derived from the Euler decomposition of
$H$. Indeed,
\[
\vcenter{\hbox{\input{diag/H-loop.tikz}}}
\, = \, \vcenter{\hbox{\input{diag/H-loop-Euler.tikz}}}
\, = \, \vcenter{\hbox{\input{diag/H-loop-Euler-2.tikz}}}
\, = \, \vcenter{\hbox{\input{diag/H-loop-Euler-3.tikz}}}
\, = \, \vcenter{\hbox{\begin{tikzpicture}
	\begin{pgfonlayer}{nodelayer}
		\node [style=empty] (0) at (-0.5, 2.5) {};
		\node [style=empty] (1) at (-0.5, 1.5) {};
		\node [style=greendot] (2) at (-0.5, 2) {$\pi$};
	\end{pgfonlayer}
	\begin{pgfonlayer}{edgelayer}
		\draw (2) to (1);
		\draw (2) to (0);
	\end{pgfonlayer}
\end{tikzpicture}}}.
\] 

In the following, we prove that Equation~\eqref{hloop} is actually strictly weaker than the Euler decomposition in the sense that the Euler decomposition cannot be derived from Equation~\eqref{hloop} in the \zxcalculus. 

\begin{lemma}\label{lem:no-euler-in-pivot}
The Euler decomposition of $H$ cannot be derived  in the \zxcalculus augmented with the axiom $\vcenter{\hbox{\begin{tikzpicture}
	\begin{pgfonlayer}{nodelayer}
		\node [style=empty] (0) at (-0.5, 2.5) {};
		\node [style=empty] (1) at (-0.5, 1.5) {};
		\node [style=greendot] (2) at (-0.5, 2) {$\pi$};
	\end{pgfonlayer}
	\begin{pgfonlayer}{edgelayer}
		\draw (2) to (1);
		\draw (2) to (0);
	\end{pgfonlayer}
\end{tikzpicture}}} 
= 
\vcenter{\hbox{\input{diag/H-loop.tikz}}}$.
\end{lemma}

\begin{proof}We consider the following functor $\denote.^\flat$ which maps diagrams to diagrams: 
\[
\begin{array}{ccccccc}
\ldenote{~~\vcenter{\hbox{\begin{tikzpicture}
	\begin{pgfonlayer}{nodelayer}
		\node [style=empty] (0) at (-1, 2) {};
		\node [style=empty] (1) at (-1, 1) {};
	\end{pgfonlayer}
	\begin{pgfonlayer}{edgelayer}
		\draw (0) to (1);
	\end{pgfonlayer}
\end{tikzpicture}}}~ }^\flat& = &~~~\vcenter{\hbox{\begin{tikzpicture}
	\begin{pgfonlayer}{nodelayer}
		\node [style=empty] (0) at (-1, 2) {};
		\node [style=empty] (1) at (-1, 1) {};
	\end{pgfonlayer}
	\begin{pgfonlayer}{edgelayer}
		\draw (0) to (1);
	\end{pgfonlayer}
\end{tikzpicture}}}~~~\vcenter{\hbox{\begin{tikzpicture}
	\begin{pgfonlayer}{nodelayer}
		\node [style=empty] (0) at (-1, 2) {};
		\node [style=empty] (1) at (-1, 1) {};
	\end{pgfonlayer}
	\begin{pgfonlayer}{edgelayer}
		\draw (0) to (1);
	\end{pgfonlayer}
\end{tikzpicture}}}&~~~~~~~~~~~~~&\ldenote{~\vcenter{\hbox{\begin{tikzpicture}
	\begin{pgfonlayer}{nodelayer}
		\node [style=empty] (0) at (-1, 2.25) {};
		\node [style=empty] (1) at (-1, 1.25) {};
		\node [style=box] (2) at (-1, 1.75) {};
	\end{pgfonlayer}
	\begin{pgfonlayer}{edgelayer}
		\draw (0) to (2);
		\draw (2) to (1);
	\end{pgfonlayer}
\end{tikzpicture}}} }^\flat& = &~~~\vcenter{\hbox{\input{diag/swap.tikz}}}\\\\
\ldenote{~\vcenter{\hbox{\begin{tikzpicture}
	\begin{pgfonlayer}{nodelayer}
		\node [style=empty] (0) at (-1, 2) {};
		\node [style=greenrot] (1) at (-1, 1.5) {};
	\end{pgfonlayer}
	\begin{pgfonlayer}{edgelayer}
		\draw (0) to (1);
	\end{pgfonlayer}
\end{tikzpicture}}} }^\flat& = &~~~\vcenter{\hbox{\begin{tikzpicture}
	\begin{pgfonlayer}{nodelayer}
		\node [style=empty] (0) at (-1, 2) {};
		\node [style=greenrot] (1) at (-1, 1.5) {};
	\end{pgfonlayer}
	\begin{pgfonlayer}{edgelayer}
		\draw (0) to (1);
	\end{pgfonlayer}
\end{tikzpicture}}}~~\vcenter{\hbox{\begin{tikzpicture}
	\begin{pgfonlayer}{nodelayer}
		\node [style=empty] (0) at (-1, 2) {};
		\node [style=reddot] (1) at (-1, 1.5) {};
	\end{pgfonlayer}
	\begin{pgfonlayer}{edgelayer}
		\draw (0) to (1);
	\end{pgfonlayer}
\end{tikzpicture}}}&~~~~~~~~~~~~~&\ldenote{~\vcenter{\hbox{\begin{tikzpicture}
	\begin{pgfonlayer}{nodelayer}
		\node [style=empty] (0) at (-1, 2) {};
		\node [style=reddot] (1) at (-1, 1.5) {};
	\end{pgfonlayer}
	\begin{pgfonlayer}{edgelayer}
		\draw (0) to (1);
	\end{pgfonlayer}
\end{tikzpicture}}} }^\flat& = &~~~\vcenter{\hbox{\begin{tikzpicture}
	\begin{pgfonlayer}{nodelayer}
		\node [style=empty] (0) at (-1, 2) {};
		\node [style=reddot] (1) at (-1, 1.5) {};
	\end{pgfonlayer}
	\begin{pgfonlayer}{edgelayer}
		\draw (0) to (1);
	\end{pgfonlayer}
\end{tikzpicture}}}~~\vcenter{\hbox{\begin{tikzpicture}
	\begin{pgfonlayer}{nodelayer}
		\node [style=empty] (0) at (-1, 2) {};
		\node [style=greenrot] (1) at (-1, 1.5) {};
	\end{pgfonlayer}
	\begin{pgfonlayer}{edgelayer}
		\draw (0) to (1);
	\end{pgfonlayer}
\end{tikzpicture}}}\\\\
\ldenote{~\vcenter{\hbox{\input{diag/greendelta.tikz}}} }^\flat& = &~~~\vcenter{\hbox{\input{diag/greenreddelta.tikz}}}&~~~~~~~~~~~~~&\ldenote{~\vcenter{\hbox{\input{diag/reddelta.tikz}}} }^\flat& = &~~~\vcenter{\hbox{\input{diag/redgreendelta.tikz}}}\\\\
\ldenote{~\vcenter{\hbox{\begin{tikzpicture}
	\begin{pgfonlayer}{nodelayer}
		\node [style=empty] (0) at (-0.5, 2.5) {};
		\node [style=empty] (1) at (-0.5, 1.5) {};
		\node [style=greendot] (2) at (-0.5, 2) {$\alpha$};
	\end{pgfonlayer}
	\begin{pgfonlayer}{edgelayer}
		\draw (2) to (1);
		\draw (2) to (0);
	\end{pgfonlayer}
\end{tikzpicture}}} }^\flat& = &~~~\vcenter{\hbox{\begin{tikzpicture}
	\begin{pgfonlayer}{nodelayer}
		\node [style=empty] (0) at (-1, 2) {};
		\node [style=empty] (1) at (-1, 1) {};
	\end{pgfonlayer}
	\begin{pgfonlayer}{edgelayer}
		\draw (0) to (1);
	\end{pgfonlayer}
\end{tikzpicture}}}~~~\vcenter{\hbox{\begin{tikzpicture}
	\begin{pgfonlayer}{nodelayer}
		\node [style=empty] (0) at (-1, 2) {};
		\node [style=empty] (1) at (-1, 1) {};
	\end{pgfonlayer}
	\begin{pgfonlayer}{edgelayer}
		\draw (0) to (1);
	\end{pgfonlayer}
\end{tikzpicture}}}&~~~~~~~~~~~~~&\ldenote{~\vcenter{\hbox{\begin{tikzpicture}
	\begin{pgfonlayer}{nodelayer}
		\node [style=empty] (0) at (-0.5, 2.5) {};
		\node [style=empty] (1) at (-0.5, 1.5) {};
		\node [style=reddot] (2) at (-0.5, 2) {$\alpha$};
	\end{pgfonlayer}
	\begin{pgfonlayer}{edgelayer}
		\draw (2) to (1);
		\draw (2) to (0);
	\end{pgfonlayer}
\end{tikzpicture}}} }^\flat& = &~~~\vcenter{\hbox{\begin{tikzpicture}
	\begin{pgfonlayer}{nodelayer}
		\node [style=empty] (0) at (-1, 2) {};
		\node [style=empty] (1) at (-1, 1) {};
	\end{pgfonlayer}
	\begin{pgfonlayer}{edgelayer}
		\draw (0) to (1);
	\end{pgfonlayer}
\end{tikzpicture}}}~~~\vcenter{\hbox{\begin{tikzpicture}
	\begin{pgfonlayer}{nodelayer}
		\node [style=empty] (0) at (-1, 2) {};
		\node [style=empty] (1) at (-1, 1) {};
	\end{pgfonlayer}
	\begin{pgfonlayer}{edgelayer}
		\draw (0) to (1);
	\end{pgfonlayer}
\end{tikzpicture}}}
\end{array}
\]
Notice that the axioms of the \zxcalculus are satisfied. Indeed, for
diagrams without $H$, the functor $\denote.^\flat$ consists in
doubling the picture and trivialising the rotations.  Regarding the
axioms which involve $H$, we have 
\[
\ldenote{~~\vcenter{\hbox{\begin{tikzpicture}
	\begin{pgfonlayer}{nodelayer}
		\node [style=empty] (0) at (-1, 2.25) {};
		\node [style=empty] (1) at (-1, 1.25) {};
		\node [style=box] (2) at (-1, 1.5) {};
		\node [style=box] (3) at (-1, 2) {};
	\end{pgfonlayer}
	\begin{pgfonlayer}{edgelayer}
		\draw (2) to (1);
		\draw (3) to (0);
		\draw (2) to (3);
	\end{pgfonlayer}
\end{tikzpicture}}}~ }^\flat
={~\vcenter{\hbox{\input{diag/2swap.tikz}}} } 
=\vcenter{\hbox{\begin{tikzpicture}
	\begin{pgfonlayer}{nodelayer}
		\node [style=empty] (0) at (-1, 2) {};
		\node [style=empty] (1) at (-1, 1) {};
	\end{pgfonlayer}
	\begin{pgfonlayer}{edgelayer}
		\draw (0) to (1);
	\end{pgfonlayer}
\end{tikzpicture}}}~~~\vcenter{\hbox{\begin{tikzpicture}
	\begin{pgfonlayer}{nodelayer}
		\node [style=empty] (0) at (-1, 2) {};
		\node [style=empty] (1) at (-1, 1) {};
	\end{pgfonlayer}
	\begin{pgfonlayer}{edgelayer}
		\draw (0) to (1);
	\end{pgfonlayer}
\end{tikzpicture}}}=\ldenote{~~\vcenter{\hbox{\begin{tikzpicture}
	\begin{pgfonlayer}{nodelayer}
		\node [style=empty] (0) at (-1, 2) {};
		\node [style=empty] (1) at (-1, 1) {};
	\end{pgfonlayer}
	\begin{pgfonlayer}{edgelayer}
		\draw (0) to (1);
	\end{pgfonlayer}
\end{tikzpicture}}}~}^\flat 
\quad \text{ and } \quad 
\ldenote{~~\vcenter{\hbox{\begin{tikzpicture}
	\begin{pgfonlayer}{nodelayer}
		\node [style=empty] (0) at (-1, 2) {};
		\node [style=box] (1) at (-1, 1.75) {};
		\node [style=reddot] (2) at (-1, 1.25) {};
	\end{pgfonlayer}
	\begin{pgfonlayer}{edgelayer}
		\draw (0) to (1);
		\draw (1) to (2);
	\end{pgfonlayer}
\end{tikzpicture}}}~}^\flat 
={~\vcenter{\hbox{\input{diag/swapreddot.tikz}}} }
=~~~\vcenter{\hbox{\begin{tikzpicture}
	\begin{pgfonlayer}{nodelayer}
		\node [style=empty] (0) at (-1, 2) {};
		\node [style=greenrot] (1) at (-1, 1.5) {};
	\end{pgfonlayer}
	\begin{pgfonlayer}{edgelayer}
		\draw (0) to (1);
	\end{pgfonlayer}
\end{tikzpicture}}}~~\vcenter{\hbox{\begin{tikzpicture}
	\begin{pgfonlayer}{nodelayer}
		\node [style=empty] (0) at (-1, 2) {};
		\node [style=reddot] (1) at (-1, 1.5) {};
	\end{pgfonlayer}
	\begin{pgfonlayer}{edgelayer}
		\draw (0) to (1);
	\end{pgfonlayer}
\end{tikzpicture}}}
=\ldenote{~\vcenter{\hbox{\begin{tikzpicture}
	\begin{pgfonlayer}{nodelayer}
		\node [style=empty] (0) at (-1, 2) {};
		\node [style=greenrot] (1) at (-1, 1.5) {};
	\end{pgfonlayer}
	\begin{pgfonlayer}{edgelayer}
		\draw (0) to (1);
	\end{pgfonlayer}
\end{tikzpicture}}}}^\flat
\] 
for instance, the other ones are satisfied similarly.

The H-loop axiom is satisfied as well:  
\[
\ldenote{\vcenter{\hbox{\input{diag/new-rule.tikz}}}}^\flat 
\, = \, \vcenter{\hbox{\input{diag/double-H-loop.tikz}}} 
\, = \, \vcenter{\hbox{\input{diag/double-H-loop-2.tikz}}} 
\, = \, \vcenter{\hbox{\begin{tikzpicture}
	\begin{pgfonlayer}{nodelayer}
		\node [style=empty] (0) at (-1, 2) {};
		\node [style=empty] (1) at (-1, 1) {};
	\end{pgfonlayer}
	\begin{pgfonlayer}{edgelayer}
		\draw (0) to (1);
	\end{pgfonlayer}
\end{tikzpicture}}}~~\vcenter{\hbox{\begin{tikzpicture}
	\begin{pgfonlayer}{nodelayer}
		\node [style=empty] (0) at (-1, 2) {};
		\node [style=empty] (1) at (-1, 1) {};
	\end{pgfonlayer}
	\begin{pgfonlayer}{edgelayer}
		\draw (0) to (1);
	\end{pgfonlayer}
\end{tikzpicture}}}
\, = \, \ldenote{\vcenter{\hbox{\begin{tikzpicture}
	\begin{pgfonlayer}{nodelayer}
		\node [style=empty] (0) at (-0.5, 2.5) {};
		\node [style=empty] (1) at (-0.5, 1.5) {};
		\node [style=greendot] (2) at (-0.5, 2) {$\pi$};
	\end{pgfonlayer}
	\begin{pgfonlayer}{edgelayer}
		\draw (2) to (1);
		\draw (2) to (0);
	\end{pgfonlayer}
\end{tikzpicture}}}}^\flat\,,
\] 
but the Euler decomposition is not:
\[
 \ldenote{\vcenter{\hbox{\input{diag/axiomEuler-r.tex}}} }^\flat 
\, = \,
 \vcenter{\hbox{\begin{tikzpicture}
	\begin{pgfonlayer}{nodelayer}
		\node [style=empty] (0) at (-1, 2) {};
		\node [style=empty] (1) at (-1, 1) {};
	\end{pgfonlayer}
	\begin{pgfonlayer}{edgelayer}
		\draw (0) to (1);
	\end{pgfonlayer}
\end{tikzpicture}}}~~\vcenter{\hbox{\begin{tikzpicture}
	\begin{pgfonlayer}{nodelayer}
		\node [style=empty] (0) at (-1, 2) {};
		\node [style=empty] (1) at (-1, 1) {};
	\end{pgfonlayer}
	\begin{pgfonlayer}{edgelayer}
		\draw (0) to (1);
	\end{pgfonlayer}
\end{tikzpicture}}}
\quad\text{ and } \quad
\ldenote{~\vcenter{\hbox{\begin{tikzpicture}
	\begin{pgfonlayer}{nodelayer}
		\node [style=empty] (0) at (-1, 2.25) {};
		\node [style=empty] (1) at (-1, 1.25) {};
		\node [style=box] (2) at (-1, 1.75) {};
	\end{pgfonlayer}
	\begin{pgfonlayer}{edgelayer}
		\draw (0) to (2);
		\draw (2) to (1);
	\end{pgfonlayer}
\end{tikzpicture}}} }^\flat
\, = \,
\vcenter{\hbox{\input{diag/swap.tikz}}}\,.
\] 
\end{proof}

%

The combination of Lemmas~\ref{lem:no-pivot-in-zx} and
\ref{lem:no-euler-in-pivot} proves that ``\zxcalculus + H-loop'' is
indeed an intermediate theory between the \zxcalculus and
``\zxcalculus + Euler''.

\section{Angle-free calculus for Real Stabilizers}
\label{sec:angle-free-calculus}

Backens \cite{Backens:2012fk} considered a syntactic
restriction on the terms of the \zxcalculus:  by demanding that all
the phases occurring in a term are multiples of $\frac{\pi}{2}$, the
resulting \zxcalculus terms are in exact correspondence with stabilizer
states.  Furthermore, the theory of \zxcalculus + Euler is
sufficient to decide the equality for these states.  In other words,
the theory is complete for stabilizer quantum mechanics.

We will now consider a stronger syntactic restriction, namely that all
phases must be either 0 or $\pi$.  Semantically this yields the
real-valued fragment of stabilizer quantum mechanics.  We will also
modify the axiom scheme by dropping the axioms ($\pi$) and (C) and
replacing them with 
\[
\begin{array}{rclcrcl}
\RULEpiBIScopyinglhs &=& \RULEpiBIScopyingrhs
&\qquad\qquad&
\RULEpicopyinglhs &=& \RULEpicopyingrhs \\
& \text{(C1)} &&&& \text{(C2)} 
\end{array}
\]
Note that these equations are both derivable in the full \zxcalculus.
The resulting system we call the weak \zxcalculus.

\begin{lemma}  \label{lem:weak-pi-commutes}
The following equations are derivable in the weak \zxcalculus:
\[
\EQNweakpihomomorphismlhs = \EQNweakpihomomorphismrhs
\qquad\qquad\quad
\EQNweakpicommutelhs = \EQNweakpicommuterhs
\]
\begin{proof}
  For the first equation we have
\[
\EQNweakpihomomorphismlhs 
= \,\PFEweakpihomomorphismi
= \,\PFEweakpihomomorphismii
= \,\PFEweakpihomomorphismiii
= \,\EQNweakpihomomorphismrhs 
\]
by spider, (B), (C2), and spider.  Making use of this equation, in
both its original and colour-switched form,  we
have:
\[
\EQNweakpicommutelhs
= \, \PFweakpicommuterhsi
= \, \PFweakpicommuterhsii
= \, \PFweakpicommuterhsiii
= \, \PFweakpicommuterhsiv
= \, \EQNweakpicommuterhs
\]
where the scalar factor was dropped at the last step.
\end{proof}
\end{lemma}

However, in the presence of Equation~\eqref{hloop} there is no need
for the angle $\pi$ at all. We can now define the ``angle-free
\zxcalculus'' by replacing all the $\pi$ vertices with loops:
\[
\vcenter{\hbox{\begin{tikzpicture}
	\begin{pgfonlayer}{nodelayer}
		\node [style=empty] (0) at (-0.5, 2.5) {};
		\node [style=empty] (1) at (-0.5, 1.5) {};
		\node [style=greendot] (2) at (-0.5, 2) {$\pi$};
	\end{pgfonlayer}
	\begin{pgfonlayer}{edgelayer}
		\draw (2) to (1);
		\draw (2) to (0);
	\end{pgfonlayer}
\end{tikzpicture}}}
\, \mapsto \,
\vcenter{\hbox{\input{diag/H-loop.tikz}}}
\qquad\text{ and }\qquad
\vcenter{\hbox{\begin{tikzpicture}
	\begin{pgfonlayer}{nodelayer}
		\node [style=empty] (0) at (-0.5, 2.5) {};
		\node [style=empty] (1) at (-0.5, 1.5) {};
		\node [style=reddot] (2) at (-0.5, 2) {$\pi$};
	\end{pgfonlayer}
	\begin{pgfonlayer}{edgelayer}
		\draw (2) to (1);
		\draw (2) to (0);
	\end{pgfonlayer}
\end{tikzpicture}}}
\, \mapsto \,
\vcenter{\hbox{\input{diag/red-H-loop.tikz}}}
\]
and replacing axiom (C2) with (L):

\begin{equation}\label{eq:loop-copy}\tag{L}
\vcenter{\hbox{\input{diag/loop-copy-lhs.tikz}}}
\, = \,
\vcenter{\hbox{\input{diag/loop-copy-rhs.tikz}}}
\end{equation}

Evidently, the resulting calculus is strictly stronger than the
weak \zxcalculus and weaker than the restricted \zxcalculus + Euler
considered by Backens.

We will show that the angle-free \zxcalculus is complete for
real-valued stabilizers.

\subsection{Real stabilizer quantum mechanics}

Recall that the Clifford operations are the normalisers of the Pauli
operators, i.e. $C_n = \{U|\forall g \in P_n, UgU^\dagger \in P_n\}$
where $P_n$ is the Pauli group on $n$ qubits.  The \emph{real}
Cliffords---i.e. those satisfying $\overline U = U$---form a subgroup
of $C_n$ generated by $\{Z,H, \CZ\}$. We call \emph{real stabilizer
  quantum mechanics} any quantum evolution that can be described by
real Clifford operations, $\ket 0$ initialisations,
and $\ket0$ projections. Notice that the image of the
angle-free \zxcalculus under the functor $\denote \cdot$ coincides with
real stabiliser quantum mechanics. We now show the
completeness of the angle-free \zxcalculus for real stabiliser
quantum mechanics, i.e. for any two diagrams $D_1$ and $D_2$, if
$\denote{D_1} = \denote{D_2}$ then $D_1=D_2$ can be proved in the
calculus.

We follow the proof of the completeness of
the \zxcalculus together with the Euler decomposition for 
(complex) stabiliser quantum mechanics \cite{Backens:2012fk}. Due to the Choi-Jamoilkowski isomorphism it suffices to
consider input-free diagrams (since any input can be turned into an
output). A diagram with no input is called a \emph{diagram state}.

\begin{definition}
  A diagram is called a \emph{GS-RLC diagram} if it consists of a
  graph state with arbitrary single real Clifford operator applied on
  each output.
\end{definition}

\begin{lemma}\label{lem:all-diags-are-gs-rlc}
  Any angle-free diagram state is equal to some GS-RLC diagram within
  the angle-free \zxcalculus.
\end{lemma}

\begin{proof}[Proof Sketch]
  The proof is by induction. Intuitively, every red dot can be turned
  into a green dot using $H$; the spider rule is used to merge green dots
  connected by a wire; parallel $H$-edges are removed using the Hopf
  law. If there is a green dot which is not connected to an output,
  then either this dot is disconnected from the rest of the diagram
  and can be ignored, or the dot can be removed by pivoting
  with one of its neighbours as shown:
  \[
  \vcenter{\hbox{\rotatebox{180}{\input{diag/tri_rem.tikz}}}}
  =\vcenter{\hbox{\rotatebox{180}{\input{diag/tri_rem2.tikz}}}}
  =\vcenter{\hbox{\rotatebox{180}{\input{diag/tri_rem3.tikz}}}}
  =\vcenter{\hbox{\rotatebox{180}{\input{diag/tri_rem4.tikz}}}}
  =\vcenter{\hbox{\rotatebox{180}{\input{diag/tri_rem5.tikz}}}}
  =\vcenter{\hbox{\rotatebox{180}{\input{diag/tri_rem6.tikz}}}}
  \]
  That is, the bottom dot is removed by pivoting along one of its
  incident edges.
\end{proof}

\begin{definition}
  A \emph{reduced} GS-RLC diagram, is a GS-RLC diagram such
  that\\$(1)$ every vertex Clifford operator is one of
  $\vcenter{\hbox{\begin{tikzpicture}
	\begin{pgfonlayer}{nodelayer}
		\node [style=empty] (0) at (-1, 2) {};
		\node [style=empty] (1) at (-1, 1) {};
	\end{pgfonlayer}
	\begin{pgfonlayer}{edgelayer}
		\draw (0) to (1);
	\end{pgfonlayer}
\end{tikzpicture}}}$,
  $\vcenter{\hbox{\input{diag/ZH-loop.tikz}}}$,
  $\vcenter{\hbox{\begin{tikzpicture}
	\begin{pgfonlayer}{nodelayer}
		\node [style=empty] (0) at (-1, 2.25) {};
		\node [style=empty] (1) at (-1, 1.25) {};
		\node [style=box] (2) at (-1, 1.75) {};
	\end{pgfonlayer}
	\begin{pgfonlayer}{edgelayer}
		\draw (0) to (2);
		\draw (2) to (1);
	\end{pgfonlayer}
\end{tikzpicture}}}$ or
  $\vcenter{\hbox{\input{diag/HZH-loop.tikz}}}$ \\$(2)$ two adjacent
  vertices must not both have vertex operators that include an $H$.
\end{definition}

\begin{lemma}
Any angle-free diagram state is equal to a reduced GS-RLC diagram.
\end{lemma}

\begin{proof}[Proof Sketch]
  Any real local Clifford is a combination of $H$, $X$ and $Z$. Notice
  that using Proposition \ref{prop:stab-graph-state}, every $X$ can
  be transformed in $Z$s on its neighbours. As a consequence the
  vertex Clifford operators are either $I$, $Z$, $H$ or
  $HZ$. Moreover, if two adjacent vertices have a vertex operator
  which include an $H$, then one can do a pivoting which is consuming
  the $H$s, transforms the graph and produces $Z$ on the common
  neighbours.  
\end{proof}

Suppose that a pair of GS-RLC diagrams describe states with the same
number of qubits, that is, they have the same set of outout vertices.
Such a pair is called \emph{simplified} if there is no pair of qubits
$u$ and $v$ which are adjacent in at least one diagram and such that
$H$ is applied on $u$ but not $v$ in the first diagram, and on $v$ but
not $u$ in the second diagram.

\begin{lemma}
Any pair of angle-free diagrams of reduced GS-LRC diagrams can be simplified.
\end{lemma}

\begin{proof}[Proof Sketch]
  If there exists a pair $u$, $v$ which are adjacent in the first
  diagram such that $H$ is applied on $u$ and not on $v$, then one can
  apply a pivoting on $u,v$ in this graph. This pivoting consumes the
  $H$ on $u$ and add an $H$ on $v$. This transformation does not
  introduce nor remove $H$ on the other vertices, so this
  transformation can be applied inductively to any pair of vertices
  which do not satisfies the conditions of simplified pairs of
  GS-RLC.
\end{proof}

\begin{theorem}
Given two reduced GS-RLC diagrams $D_1$ and $D_2$ which form a simplified pair, $\denote{D_1} = \denote{D_2}$ if and only if $D_1$ and $D_2$ are identical. 
\end{theorem}

\begin{proof}
  Since $D_1$ and $D_2$ are reduced GS-RLC diagrams, there exist two graphs $G_1 = (V,E_1)$ and $G_2 = (V,E_2)$, and four
  subsets $A_1,A_2,B_1,B_2\subseteq V$ such that $\denote{D_i} =
  H_{A_i}Z_{B_i} \ket {G_i}$, where $H_A = \bigotimes_{u\in A}
  H_u$. Diagrams $D_1$ and $D_2$ are identical iff $A_1=A_2$,
  $B_1=B_2$ and $G_1=G_2$. First we show that $A_1=A_2$. Notice that
  $\denote{D_1} = \denote{D_2}$ iff $H_{A}Z_{B_1}\ket{G_1} =
  Z_{B_2}\ket{G_2}$ where $A =A_1\Delta A_2$ is the symmetric
  difference of $A_1$ and $A_2$. By contradiction, for any $u\in A_1\Delta
  A_2$,  $\ket{G_1}$ is a fix point of $X_uZ_{N_{G_1}(u)}$, so $Z_{B_1}
  \ket{G_1}$ is an eigenvector of  $X_uZ_{N_{G_1}(u)}$. Moreover,
  since $D_1$ is in a reduced form there is no $H$ applied on qubits
  adjacent to $u$, so $H_{A}Z_{B_1}\ket{G_1}$ is an eigenvector of
  $Z_uZ_{N_{G_1}(u)}$. Indeed  $H_{A}Z_{B_1}\ket{G_1} = H_{A}Z_{B_1}X_uZ_{N_{G_1}(u)}\ket{G_1}=\pm  Z_uZ_{N_{G_1}(u)}H_{A}Z_{B_1}\ket{G_1}$. Regarding the second state, $Z_{B_2}\ket{G_2}$
  is an eigenvector of $X_uZ_{N_{G_2(u)}}$. The two operator
  $X_uZ_{N_{G_2(u)}}$ and $Z_uZ_{N_{G_1}(u)}$ are anti commuting so
  they cannot have a common non-zero eigenvector, as a
  consequence $A_1 = A_2$. Thus $\denote{D_1} = \denote{D_2}$ implies
  $Z_{B_1}\ket {G_1} = Z_{B_2}\ket{G_2}$. Moreover, it has been proved (Lemma 3 in
  \cite{MP2013}) that $Z_{B_1}\ket {G_1} = Z_{B_2}\ket{G_2}$ implies  $B_1 = B_2$ and $G_1 = G_2$. As
  a consequence the two diagrams are identical.
\end{proof}

\section{Conclusion and Perspectives}

We have introduced a new calculus, intermediate between the
\zxcalculus and the \zxcalculus augmented with the Euler decomposition
of $H$. As the introduction of the Euler decomposition was driven by
local complementation, the new axiom we consider, namely the H-loop,
is driven by another graph transformation, namely pivoting. We prove
the H-loop axiom cannot be derived in the plain \zxcalculus, and is
strictly weaker than the Euler decomposition of $H$. When restricted
to $0$- and $\pi$-rotations this new calculus is complete for real
stabiliser quantum mechanics. Moreover this restricted language admits
a simple equivalent angle-free calculus. We believe this angle-free
calculus will be the cornerstone for an axiomatisation of real quantum
mechanics. Real quantum mechanics is known to be universal for quantum
computing, moreover the restriction to the real field provides some
useful simplifications in terms of diagrammatic quantum mechanics (for
example, the object $A$ and its dual $A^*$ have the same
interpretation). Another example is that, when restricted to real
numbers, the unbiased bases are perfectly captured by the
complementary observables $X$ and $Z$ of the \zxcalculus, whereas the
axiomatisation of the third (complex) mutually unbiased base for
qubit, albeit possible (see \cite{Lang:2011fk}) is less
intuitive.   On the other hand applications which require complex
numbers like local tomography cannot be captured by this intermediate
language, and require additional axioms (e.g. Euler decomposition of
$H$).



\subsubsection*{Acknowledgments. } 
 This work has been partially funded by the ANR-10-JCJC-0208 CausaQ grant. 
\bibliography{pivot}
\end{document}